\documentclass[runningheads]{llncs}

\usepackage{algorithm}
\usepackage{algorithmic}
\usepackage{amsmath}\allowdisplaybreaks
\usepackage{amssymb}
\usepackage{color,xcolor}
\usepackage[margin = 1in]{geometry}
\usepackage{hyperref}
\usepackage{ifthen}
\usepackage{mathtools}
\usepackage{braket}

\usepackage{graphicx}
\usepackage{subfigure}
\usepackage{caption}
\usepackage{enumitem}
\usepackage{float}

\usepackage[style=numeric]{biblatex}
\hypersetup{colorlinks=true, linkcolor=blue, citecolor=blue, anchorcolor=blue, urlcolor=blue}  

\newtheorem{mechanism}{Mechanism}

\newcommand{\sect}[1]{\hyperref[sect:#1]{Section~\ref*{sect:#1}}}
\newcommand{\append}[1]{\hyperref[append:#1]{Appendix~\ref*{append:#1}}}
\newcommand{\lem}[1]{\hyperref[lem:#1]{Lemma~\ref*{lem:#1}}}
\newcommand{\obs}[1]{\hyperref[obs:#1]{Observation~\ref*{obs:#1}}}
\newcommand{\thm}[1]{\hyperref[thm:#1]{Theorem~\ref*{thm:#1}}}
\newcommand{\mech}[1]{\hyperref[mech:#1]{Mechanism~\ref*{mech:#1}}}
\newcommand{\cor}[1]{\hyperref[cor:#1]{Corollary~\ref*{cor:#1}}}
\newcommand{\figg}[1]{\hyperref[fig:#1]{Figure~\ref*{fig:#1}}}
\newcommand{\tab}[1]{\hyperref[tab:#1]{Table~\ref*{tab:#1}}}
\newcommand{\ex}[1]{\hyperref[ex:#1]{Example~\ref*{ex:#1}}}
\newcommand{\defi}[1]{\hyperref[def:#1]{Definition~\ref*{def:#1}}}
\newcommand{\assump}[1]{\hyperref[assump:#1]{Assumption~\ref*{assump:#1}}}
\newfloat{mechanism}{htbp}{mechanism}
\floatname{mechanism}{Mechanism}



\def\>{\rangle}
\def\<{\langle}

\newcommand{\compilehidecomments}{false}

\ifthenelse{ \equal{\compilehidecomments}{false} }{%
	\newcommand{\zongqi}[1]{{\color{red}  [\text{Zongqi:} #1]}}
\newcommand{\wei}[1]{{\color{blue} [\text{Wei:} #1]}}
\newcommand{\jialin}[1]{{\color{cyan} [\text{Jialin:} #1]}}
\newcommand{\zhijie}[1]{{\color{yellow} [\text{Zhijie:} #1]}}        
}{
	\newcommand{\zongqi}[1]{}
\newcommand{\wei}[1]{}
\newcommand{\jialin}[1]{}
\newcommand{\zhijie}[1]{}
}

\title{Design and Characterization of Strategy-Proof Mechanisms for\newline Two-Facility Game on a Line}

\author{
Pinyan Lu
\and
Zihan Luo \and
Jialin Zhang
}
\institute{Shanghai University of Finance and Economics, China \and
Institute of Computing Technology Chinese Academy of Sciences, China
\and
Institute of Computing Technology Chinese Academy of Sciences, China}

\begin{document}
\maketitle

\begin{abstract}
    We focus on the problem of placing two facilities along a linear space to serve a group of agents. Each agent is committed to minimizing the distance between her location and the closest facility. A mechanism is an algorithm that maps the reported agent locations to the facility locations. We are interested in mechanisms without money that are deterministic, strategy-proof, and provide a bounded approximation ratio for social cost. 


    It is a fundamental problem to characterize the family of strategy-proof mechanisms with a bounded approximation ratio. Fotakis and Tzamos already demonstrated that the deterministic strategy-proof mechanisms for the 2-facility game problem are mechanisms with a unique dictator and the leftmost-rightmost mechanism. In this paper, we first present a more refined characterization of the first family. 

    We then reveal three new classes of strategy-proof mechanisms that show the intricacy of structure within this family. This helps us get a more complete picture of the characterization of the 2-facility game problem, and may also have value in understanding and solving more general facility allocation game problems.

    Besides, based on our refined characterization, we surprisingly find that prediction cannot effectively improve the performance of the mechanism in the two-facility game problem, while this methodology to overcome bad approximation ratio works in many other mechanism design problems. We show that if we require that the mechanism admits a bounded approximation ratio when the prediction is arbitrarily bad, then at the same time, the mechanism can never achieve sublinear approximation ratios even with perfect prediction.
\end{abstract}

\section{Introduction}

The facility game problem, characterizing a large number of scenarios about public resource allocation, is an important research topic in the study of social choice theory and algorithmic mechanism design \cite{lavi2011truthful,feldman2011randomized,ju2008efficiency,schummer2007mechanism,miyagawa2001locating}.
In this problem, the government requires agents in need to report their locations and then uses a mechanism to calculate the location of the facilities. 

However, the objectives of agents and the government are not the same. For agents, they only focus on minimizing their own cost, that is their distance to the nearest facility. The government focuses on minimizing the total distance from each agent to her nearest facility, which is known as social cost. Therefore, the most important requirement in mechanism design research is strategy-proof: even if each agent only cares about her own interests, honestly reporting her location is still optimal. If a mechanism is not strategy-proof, agents may act dishonestly to obtain a greater personal benefit, even if it goes against the goals of the government. This can lead to uncontrollable reported location of agents, which in turn leads to uncontrollable results and performance of the mechanism. Additionally, we also care about the performance of the mechanism, calculate the total distance from each agent to her nearest facility, and compare it with the optimal value when strategy-proofness is not considered. Undoubtedly, due to the constraints of strategy-proofness, many mechanisms cannot achieve optimal performance, and we usually use approximation ratios to measure the performance of different strategy-proof mechanisms.

In the paper, we focus on the most classical and important setting of facility game on a line. In this setting, the location of an agent $i$ can be denoted by a real number $x_i$. 

\subsection{Characterization of strategy-proof mechanisms}

Characterizing the family of strategy-proof mechanisms is one of the most important and fundamental problems in the area of social choice theory and mechanism design. The renowned Gibbard-Satterthwaite theorem \cite{sen2001another} demonstrated that in cases where preferences on alternatives for each agent can be arbitrary, the only strategy-proof mechanisms are dictatorships when the number of alternatives exceeds two. A dictatorship mechanism always selects the most preferred alternative for a specific agent $i$, who is termed the dictator. In the context of the facility game, a dictator $i$ means that we always build a facility exactly at the location of $x_i$, which naturally aligns with her preference.

However, the facility game on a line introduces a constraint: preferences on facility locations are not arbitrary. Specifically, each agent $i$ has a single preferred location $x_i$. When two facility locations are on the same side of $x_i$, agent $i$ will always favor the one closer to $x_i$. These specific preferences are referred to as single-peaked preferences, a concept first explored by Black \cite{black1948rationale}. Because of this constraint, the Gibbard-Satterthwaite theorem does not hold with single-peaked preferences. Consequently, the facility game allows for a broader range of strategy-proof mechanisms. Moulin \cite{moulin1980strategy}  characterized the class of all strategy-proof mechanisms for a one-facility game on the real line. Notably, Barbera and Jackson \cite{barbera2004choosing}, as well as Sprumont \cite{sprumont1991division}, refined Moulin's work by eliminating an unnecessary assumption in the proof. Particularly, for the one-facility game, a generalized median voter scheme is sufficient to characterize all strategy-proof mechanisms. For an in-depth exploration, readers can consult Barbera's comprehensive survey \cite{barbera2001introduction}. Furthermore, it is worth mentioning that the median mechanism stands out as 1-approximation, consistently achieving the optimal social cost.

In this paper, our focus centers on characterizing the two-facility game. The structure of strategy-proof mechanisms for this game is notably richer than that of the one-facility game. This complexity arises from the substantial tie-ness within agents' preferences. Agents only prioritize the closer facility, disregarding the position of the further one, resulting in all those outputs being tied from the agent's viewpoint. Specifically, the following mechanism is strategy-proof: the first facility precisely aligns with agent $i$'s position $x_i$, and the second facility's location is an arbitrary function of agent $i$'s reported position. However, this mechanism is somewhat uninteresting as it underutilizes the second facility. To explore more compelling mechanisms, we adopt the framework from \cite{lu2010asymptotically} which focuses on mechanisms with bounded approximation ratios. The previously mentioned mechanism fails to satisfy this crucial property.

As an extension of the generalized median scheme, the following mechanism is strategy-proof: placing two facilities at the positions of the $k$-th and $j$-th agents, for any given $k$ and $j$. It is evident that among these options, the only one with a bounded approximation ratio is positioning the facilities at the leftmost and rightmost positions—namely, the 1st and n-th agents.

The characterization results outlined in \cite{fotakis2014power} demonstrate that the deterministic strategy-proof mechanisms with bounded approximation ratios for the 2-facility game problem are mechanisms with a unique dictator and the leftmost-rightmost mechanism. While the second mechanism is a concrete solution and thus a complete characterization, the same cannot be said for the first family. This is because it only specifies placing one facility with a dictator, leaving the allocation of the other facility unspecified.

In this paper, we refine the characterization of the first family. We prove that the other facility, apart from the one located based on the dictatorship, is either positioned to the left of or coincident with the leftmost agent, or to the right of or coincident with the rightmost agent, or coincident with the first facility.

Our contributions unfold in two significant ways. Firstly, we present a more refined characterization of the family of mechanisms with a unique dictator. Secondly, our research reveals a remarkably intricate structure within this family. Historically, only one mechanism was recognized, where the other facility's placement was determined solely by the relative positions of the leftmost and rightmost agents. However, our study uncovers several compelling strategy-proof mechanisms, where the reported positions of multiple agents can influence the facility's location. Many of these intricacies were previously undiscovered, making these mechanisms inherently captivating. These novel findings mark a crucial step toward achieving the final, comprehensive characterization of this family.

\subsection{Mechanism design with prediction}
We introduce a distinct yet interconnected line of research: the novel field of mechanism design with prediction. Similar to the concept of algorithm design with prediction \cite{purohit2018improving,lykouris2021competitive,lattanzi2020online}, the incorporation of accurate predictions holds the potential to significantly enhance mechanism performance, a phenomenon referred to as consistency, without compromising the worst-case guarantee too much, known as robustness~\cite{purohit2018improving}. Within this paradigm, the study of two-facility games emerges as a natural focus, given the inherent challenge posed by a stringent lower bound of $\Omega(n)$ for deterministic strategy-proof mechanisms \cite{fotakis2014power}. While facility games with prediction have been explored previously \cite{agrawal2022learning,xu2022mechanism,istrate2022mechanism}, a pivotal question surfaced in \cite{xu2022mechanism}: Can a strategy-proof mechanism achieve consistency within $o(n)$ while maintaining bounded robustness?

Our research addresses this question and presents a striking and definitive finding: achieving such a balance is inherently unattainable. We formally establish the impossibility of devising a strategy-proof mechanism with $o(n)$ approximation under accurate predictions while ensuring a bounded approximation in cases of inaccurate predictions. To the best of our knowledge, this represents the inaugural strong impossibility result in this domain. Our proof rests on a refined characterization, as elaborated in the preceding subsection.

\subsection{Related works}

The facility game has a rich history in social science literature. In \cite{procaccia2013approximate}, the strategic facility game problem was first introduced as an example of a mechanism without money. 

For the single-facility game problem, current research has already achieved comprehensive results. For this problem on a line, Procaccia and Tennenholtz \cite{procaccia2013approximate} presented a 1-approximation mechanism for the maximum agent's cost and a 2-approximation mechanism for social cost, both achieving the optimal approximation ratio. While in two-dimensional Euclidean space, Meir \cite{meir2019strategyproof} and Goel and Hann-Caruthers \cite{goel2022optimality} presented the mechanisms with optimal approximation ratio for maximum agent's cost and social cost, respectively. Alon et al. \cite{alon2010strategyproof} and Meir \cite{meir2019strategyproof} studied this problem in general metric spaces, d-dimensional Euclidean spaces, circles, and trees, which all achieved constant approximation ratios.

For the facility game problem with more than two facilities, the related work has not come close to achieving near-optimal results in most cases. For the k-facility game problem, Escoffier et al. \cite{escoffier2011strategy} designed strategy-proof mechanisms to locate $k$ facilities for $k+1$ agents and achieves an approximation ratio of $\frac{k+1}{2}$ in general metric spaces. Fotakis and Tzamos \cite{fotakis2014power} proved that when $k\geq 3$, there do not exist any deterministic anonymous strategy-proof mechanisms with a bounded approximation ratio.

Therefore, many works focus on two-facility game problems on a line which is the simplest setting. For the two-facility game problem that the government considers minimizing the maximum agent's cost, Procaccia and Tennenholtz \cite{procaccia2013approximate} proposed a mechanism with a 2-approximation ratio. 
Therefore, the setting with minimizing the social cost is the most complex aspect of the facility game problem. 
In the study on the two-facility game problem on a line for social cost, Fotakis and Tzamos \cite{fotakis2014power} demonstrate a characterization result that the deterministic strategy-proof mechanisms with bounded approximation ratios for this problem are mechanisms with a unique dictator and the leftmost-rightmost mechanism. 
This corresponds to the only two mechanisms we currently know about for this problem. The first is the leftmost-rightmost mechanism proposed in Procaccia and Tennenholtz \cite{procaccia2013approximate} with an approximation ratio $n-2$.  The second is a specific mechanism with a unique dictator, proposed in Lu et al. \cite{lu2010asymptotically} with an approximation ratio of $2n-1$.  The lower bound of the deterministic strategy-proof mechanism is shown as $2$ in \cite{lu2009tighter}. Fotakis and Tzamos \cite{fotakis2014power} improved the lower bound to $n-2$, which is the same as the upper bound, and proved that the deterministic mechanism in \cite{procaccia2013approximate} is optimal. In addition, Lu, Wang, and Zhou \cite{lu2009tighter} obtained an upper bound of $4$ and a lower bound of $1.045$ for randomized, strategy-proof mechanisms.

Unlike the research on online algorithms with prediction, there has been little research on considering mechanism design with prediction. In 2022, Xu and Lu \cite{xu2022mechanism} began studying mechanism design with prediction for facility location games and designed a dictator mechanism that achieved $(1+\frac{n}{2})$-consistency and $(2n-1)$-robustness for social cost, while the optimal approximation ratio without prediction is $n-2$. In the same year, Agrawal et al. \cite{agrawal2022learning} studied a (two-dimensional) facility game problem with prediction and designed mechanisms with 1-consistency and $(1+\sqrt2)$-robustness for social cost, and $\frac{\sqrt{2c^2+2}}{1+c}$-consistency and $\frac{\sqrt{2c^2+2}}{1-c}$-robustness for maximum agent's cost. Additionally, Istrate and Bonchis \cite{istrate2022mechanism} studied obnoxious facility game problems with prediction and designed a mechanism with the optimal trade-off between consistency and robustness.

Several recent papers have investigated various facility problems.
Kanellopoulos et al. \cite{kanellopoulos2023discrete} studied the discrete heterogeneous two-facility location problem.
Xu et al. \cite{xu2021two} focused on two-facility location games with minimum distance requirements, and designed mechanisms with constant approximation ratios.
Zhou et al. \cite{zhou2023facility}  introduced lower thresholds and upper thresholds for the agent’s cost in the setting.
Deligkas et al. \cite{deligkas2023heterogeneous} initiated the study of the heterogeneous facility location problem with limited resources and designed deterministic and randomized strategy-proof mechanisms.
Chan et al. \cite{ijcai2021p596} summarized all existing mechanisms for facility location problems.

\section{Preliminaries}
Let $(\mathrm{\Omega}, d)$ be a metric space, and $d: \mathrm{\Omega}\times \mathrm{\Omega}\rightarrow  \mathnormal{R}$ is the metric. In this paper, we consider the linear metric space, where $ \mathrm{\Omega}= \mathnormal{R}$, and $d\left(x,y\right)$ is the Euclidean distance between $x,y\in \mathnormal{R}$. We define the relative position of two points $x$ and $y$ as follows: if $x<y$, we say that $x$ is on the left of $y$, while if $x>y$, we say that $x$ is on the right of $y$. 


Let the set of agents be $N=\{1,2,...,n\}$, where the location of agent $i$ is $x_i\in \mathrm{\Omega}$. We call $\mathbf{x}=(x_1,x_2,...,x_n)$ a location profile, or an instance. In particular, we define a 3-location profile, where there are three sets $N_1, N_2, N_3$ satisfy $\{N_1, N_2, N_3\}$ is a partition of the set $N$, and the location of all agents in set $N_i$ is $y_i$, $i\in\{1,2,3\}$. Let $I_3(N)=(y_1:N_1,y_2:N_2,y_3:N_3)$ denote a 3-location profile for $\{N_1, N_2, N_3\}$. 
Suppose $\pi$ arranges all the three sets according to the increasing order of their locations, i.e. $y_{\pi\left(1\right)}\le y_{\pi\left(2\right)}\le y_{\pi\left(3\right)}$, we say that $\pi$ is the permutation of $\{N_1, N_2, N_3\}$. If $I_3(N)$ is a 3-location profile for $\{N_1, N_2, N_3\}$, and the permutation of $\{N_1, N_2, N_3\}$ is $\pi$, we say that $I_3(N)$ is arranged according to the permutation $\pi$. For example, when $I_3(N)=(y_1:N_1,y_2:N_2,y_3:N_3)$, and $y_1\leq y_3\leq y_2$, we have $\pi=(1,3,2)$ is the permutation of $\{N_1, N_2, N_3\}$, and $I_3(N)$ are arranged according to the permutation $\pi$. In additionally, $\pi(1)=1,\pi(2)=3,\pi(3)=2$.

In the two-facility game on a line, a deterministic mechanism will output two facility locations for a given location profile $\mathbf{x}$, which is actually a function $f: \mathnormal{R}^n\rightarrow \mathnormal{R}^2$. The goal of each agent is to minimize her own cost, i.e., for agent $i$, she will manipulate her reported location to minimize her cost. Given a reported location profile $\mathbf{x}$ and a mechanism $f$, assuming that the facility locations output by the mechanism $f$ for $x$ are $l_1$ and $l_2$, the cost of agent $i$ is its distance from her real location $x_i$ to the nearest facility:
\[cost\left(\{l_1,l_2\},x_i\right)=min\{d\left(l_1,x_i\right),d\left(l_2,x_i\right)\}.\]
If all agents are closer to $l_1(resp.\ l_2)$ than $l_2(resp.\ l_1)$, we say that $l_2(resp.\ l_1)$ is unavailable.

The goal of the mechanism designer is to choose an appropriate mechanism that minimizes the social cost. The social cost of the mechanism $f$ on a reported location profile $\mathbf{x}$ is the total cost of all agents. When the mechanism is strategy-proof (as defined below), the real location profile for all agents is still $\mathbf{x}$:
\[SC(f,\mathbf{x})=\sum_{i=1}^{n}{cost(\{l_1,l_2\},x_{i})}.\]

For a location profile $\mathbf{x}$, the optimal social cost is denoted by $OPT(\mathbf{x})$. The approximation ratio of the mechanism $f$ is $\gamma$ means that for any location profile $\mathbf{x}\in \mathrm{\Omega}^n$, we have:
\[SC(f,\mathbf{x})\le\gamma OPT\left(\mathbf{x}\right).\]

Let $\mathbf{x}_{-i}=(x_1,...,x_{i-1},x_{i+1},...,x_n)$ be the location profile without agent $i$, and we write $\mathbf{x}=\langle x_i,\mathbf{x}_{-i}\rangle$ to emphasize the location of agent $i$. Similarly, when $S\subset N$ is a set of agents, $\mathbf{x}_{-S}$ is the location profile of agents except $S$. Let $\mathbf{x}=\langle\mathbf{x}_S,\mathbf{x}_{-S}\rangle$ represent the location profile in which the agents in $S$ report the locations ${\mathbf{x}}_S$, while the other agents report the locations $\mathbf{x}_{-S}$. For simplicity, we let $f\left(\langle x_i,\mathbf{x}_{-i}\rangle\right)=f\left(x_i,\mathbf{x}_{-i}\right), f\left(\langle x_S ,x_{-S} \rangle\right)=f\left(x_S,x_{-S}\right).$

Now, there is the formal deﬁnitions of strategy-proofness.

\begin{definition}[strategy-proofness]
   We say a mechanism $f$ is strategy-proof if no agent can benefit from misreporting her location. 
 
   Specifically, given an agent $i$, a location profile $\mathbf{x}=\langle x_i,\mathbf{x}_{-i}\rangle\in \mathnormal{R}^n$, if $i$ misreports her location from $x_i$ to $x_i^{\prime}\in \mathnormal{R}$, a strategy-proof mechanism always satisfies:
  \[cost\left(f\left(x_i,\mathbf{x}_{-i}\right),x_i\right)\le cost\left(f\left(x_i^\prime,\mathbf{x}_{-i}\right),x_i\right).\]
\end{definition}

In mechanism design problems, we tend to only consider the deterministic, strategy-proof mechanisms with bounded approximation ratio, and we denote mechanisms that satisfy these conditions as \textit{nice mechanisms}.

In this paper all mechanisms we discuss are scale-free. Given an arbitrary location profile $\mathbf{x}=(x_1, x_2, ..., x_n)$ and a mechanism $f$, assuming that $f(\mathbf{x})=\{l_1, l_2\}$, if for all $\mathbf{x^\prime}=(a{x_1}+c, a{x_2}+c, ..., a{x_n}+c)$ where $a>0$ and $c\in \mathnormal{R}$, we have $f(\mathbf{x^\prime})=\{a{l_1}+c, a{l_2}+c\}$, we define $f$ as a scale-free mechanism. Without loss of generality, if a mechanism is scale-free, we can proportionally change the location profile reported by agents, without changing the relative location of facilities. We call a location profile $\mathbf{x}$ normalized if the location of the leftmost agent is 0 and the location of the rightmost agent is 1.

\section{Property of Nice Mechanisms for Two-facility Game on a Line}

We begin by demonstrating that any nice mechanism for the two-facility game on a straight line must adhere to a specific property: either at least one of the facility locations output by the mechanism is no greater than or no less than the reported location of all agents, or both facility locations are the same. 

This property is not a complete characterization that can characterize all strategy-proof mechanisms. However, it provides a detailed characterization of the mechanism's output, which has considerable value.
This property is the basis of various strategy-proof mechanisms we will show later, which 
can deepen our understanding of this problem to fully characterize its strategy-proof mechanisms. Furthermore, this property allows us to establish a linear lower bound for the consistency of this problem with prediction, this indicates that the problem does not produce sublinear outcomes even with prediction.


\begin{theorem}
 If $f$ is a nice mechanism for the two-facility game on a straight line applied to $n \geq 5$ agents, for arbitrary instances $\mathbf{x}$, it must satisfy the following property: assuming that $f(\mathbf{x}) = \{l_1, l_2\}$, either there exists $l \in f(\mathbf{x})$ such that $l \geq \max \mathbf{x}$ or $l \leq \min \mathbf{x}$, or $l_1=l_2$.
\end{theorem}

We will begin by introducing three lemmas that will be useful in proving Theorem 1.

\begin{lemma}[Theorem 4.1 \cite{fotakis2014power}]
If $f$ is a nice mechanism for the two-facility game applied to $n \geq 5$ agents, then either for all instances $\mathbf{x}$, $f(\mathbf{x}) = \{\min \mathbf{x}, \max \mathbf{x}\}$ or there exists a unique dictator $j$ such that for all $\mathbf{x}$, $x_j \in f(\mathbf{x})$.
\end{lemma}

\begin{lemma}[Corollary 3.2 \cite{lu2010asymptotically}]
    Assuming that $f$ is a strategy-proof mechanism for two-facility game, with a location profile $\mathbf{x} = \langle x_i, \mathbf{x}_{-i}\rangle$, if $z \in f(\mathbf{x})$, then it must also be the case that $z \in f(z, \mathbf{x}_{-i})$.
\end{lemma}

\begin{proof}
    If Lemma 2 is not satisfied, that is, there exists a location profile $\mathbf{x} = \langle x_i, \mathbf{x}_{-i}\rangle$, such that $z \in f(\mathbf{x})$ but $z \notin f(z, \mathbf{x}_{-i})$.
    In this case, agent $i$ can gain a benefit by misreporting her location as $x_i$ when her true location is $z$, the mechanism $f$ is not strategy-proof, which comes to a contradiction.
\end{proof}


\begin{lemma}[Corollary 6.1 \cite{fotakis2014power}]
        Suppose $f$ is a nice mechanism applied to any 3-location profile with $n\geq 3$ agents. Given three sets $\{N_1, N_2, N_3\}$, which is a partition of $N$. At most two permutations $\pi_1,\ {\pi}_2$ for $\{N_1, N_2, N_3\}$ exist such that $\pi_1\left(2\right)=\pi_2\left(2\right)$, for all 3-location profile $\mathbf{x}$ for $\{N_1,N_2,N_3\}$, assuming that $f(\mathbf{x})$ places a facility at the middle set, then $\mathbf{x}$ arranged according to either permutation $\pi_1$ or $\pi_2$. When $\mathbf{x}$ is arranged according to other permutations on the line, $f(\mathbf{x})=\{\min \mathbf{x},\max \mathbf{x}\}$.
\end{lemma}

We will now commence with the proof of Theorem 1.

\begin{proof}
    Lemma 1 implies that a nice mechanism $f$ for the two-facility game on a straight line has only two possibilities. If for all instances $\mathbf{x}$, $f(\mathbf{x}) = \{\min \mathbf{x}, \max \mathbf{x}\}$, then both facility locations satisfy $l \geq \max \mathbf{x}$ or $l \leq \min \mathbf{x}$, so the theorem holds. If $f$ admits a unique dictator $j$, we denote the location of the facility other than $x_j$ that $f$ outputs as $l_2$. We only need to prove that either $l_2 \geq \max \mathbf{x}$ or $l_2 \leq \min \mathbf{x}$, or $l_2 = x_j$.

    To prove that, we will use proof by contradiction. We assume there exists a nice mechanism $f$ which admits a unique dictator $j$, and for a location profile $\mathbf{x}$, $l_2$ satisfies $\min \mathbf{x} < l_2 < \max \mathbf{x}$ and $l_2 \neq x_j$. Then we will demonstrate that this leads to a contradiction.
    
    For a location profile $\mathbf{x}=(x_1,x_2,...,x_n)$, we define the set of agents on the leftmost location as $S_l$, and the set of agents on the rightmost location as $S_r$. Without loss of generality, we assume that $\min \mathbf{x} \leq x_j < l_2 < \max \mathbf{x}$. 

    We will move agents according to Lemma 2 and create a 3-location profile. Lemma 2 claims that when we move an agent to $l_2$ except for $S_r$ and $j$, the facility location $l_2$ is unchanged.
    By iterative moving all agents except for $S_r$ and $j$ to the location $l_2$, we obtain $\mathbf{x}' = (x_j:j,\ l_2:N/\{j,S_r\},\ \max \mathbf{x}:S_r)$. According to Lemma 2, $l_2 \in f(\mathbf{x}')$. Each location profile $\mathbf{x}$ corresponds to a unique location profile $\mathbf{x}'$, as shown in Figure 1.

    \begin{figure}[H]
    \centering
    \includegraphics[width=0.5\textwidth]{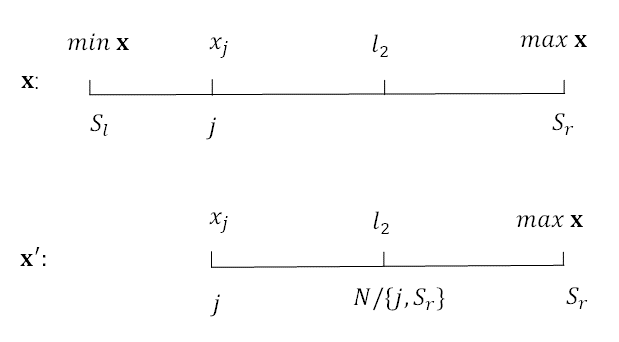}
    \caption{The correspondence between $\mathbf{x}$ and $\mathbf{x}^\prime$. The positions labeled above the line segment are the locations of the points, and below are the agents or the sets of agents located at those positions. There is a possibility that min $\mathbf{x}$ and $x_j$ overlap, in which case $j \in S_l$.}
    \label{figure3-1}
    \end{figure}


    Given three sets $N_1={j}, N_2=N/(j,S_r), N_3={S_r}$, the location profile $\mathbf{x}^\prime$ we obtained above is a 3-location profile for $\{N_1,N_2,N_3\}$ with $l_2\in f(\mathbf{x}^\prime)$, and we know that $\mathbf{x}^\prime$ is arranged according to $\pi_1=(1, 2, 3)$. The dictator $j$ is located at position $x_j$ in $\mathbf{x}^\prime$. We exchange the locations of the agent $j$ and the agents set $N/(j,S_r)$, resulting in a new location profile $\mathbf{x}^{\prime\prime}=(l_2:j,x_j:N/(j,S_r),\max \mathbf{x}:S_r)$, as shown in Figure 2. We observed that $\mathbf{x}^{\prime\prime}$ is also a 3-location profile for $\{N_1,N_2,N_3\}$, and is arranged according to $\pi_2=(2,1,3)$. Since $j$ is the unique dictator admitted by mechanism $f$, in the location profile $\mathbf{x}^{\prime\prime}$, $j$ is located at $l_2$, so $l_2\in f(\mathbf{x}^{\prime\prime})$. Thus, we can arrive that $\pi_1(2)=2, \pi_2(2)=1, \pi_1(2)\neq\pi_2(2)$, and we observe that mechanism $f$ places a facility at the middle set for both $\pi_1$ and $\pi_2$. Lemma 3 claimed that given $\{N_1, N_2, N_3\}$, a nice mechanism $f$ will place a facility at the middle set for $\pi_1$ and $\pi_2$ if and only if $\pi_1(2)=\pi_2(2)$. This indicates that $f$ is not a nice mechanism, which leads to a contradiction.

    \begin{figure}[H]
    \centering    \includegraphics[width=0.5\textwidth]{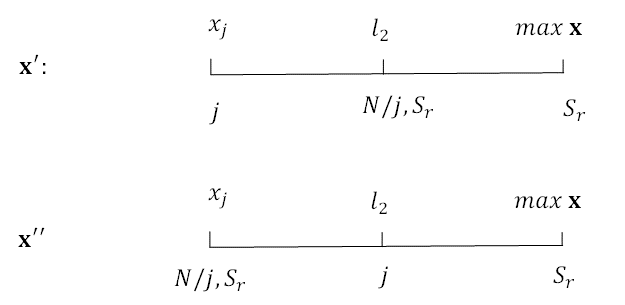}
    \caption{The correspondence between $\mathbf{x}^\prime$ and $\mathbf{x}^{\prime\prime}$. The positions labeled above the line segment are the locations of the points, and below are the agents or the sets of agents located at those positions.}
    \label{figure3-2}
    \end{figure}

    Therefore, there does not exist a nice mechanism $f$ that admits a unique dictator $j$,  such that for location profile $\mathbf{x}$, $\min \mathbf{x} \leq l_2 \leq \max \mathbf{x}$ and $l_2 \neq x_j$. That is, for all nice mechanisms $f$ that recognize a unique dictator $j$, given a location profile $\mathbf{x}$, either $l_2 \geq \max \mathbf{x}$ or $l_2 \leq \min \mathbf{x}$, or $l_2=x_j$.

    Theorem 1 is proven.
\end{proof}

\section{Two-Facility Game with Prediction}

We consider the prediction model which can directly predict the location profile of all agents. The prediction is obtained through an oracle (possibly some machine learning method) and cannot be manipulated by agents and the government. Mechanisms with prediction for this problem are a special kind of mechanism for this problem, so they must satisfy the characterization proved above.
Based on this characterization, we find that accurate prediction cannot effectively improve the consistency of the mechanism in the two-facility game problem. 
This is a strong negative result, indicating the mechanism can never achieve sublinear approximation ratios, even with accurate prediction of the locations of all agents.

In this model, the prediction is the location profile $\widetilde{\mathbf{x}}=({\widetilde{x}}_1,...,{\widetilde{x}}_n)$ given by an oracle, where ${\widetilde{x}}_i$ is the predicted location of agent $i$. The prediction is public information, and cannot be manipulated by agents and the government. For a mechanism $f$ with prediction, its input is the prediction and the reported location profile, and output is two facility locations. When the prediction is completely accurate, in other words, the prediction $\widetilde{\mathbf{x}}$ is exactly the real location profile of all agents, at which time the approximation ratio of the mechanism $f$ is no more than $\gamma$, then we define $\gamma$ as the consistency of mechanism $f$. For prediction $\widetilde{\mathbf{x}}$ with arbitrary any error, the approximation ratio of the mechanism $f$ is always no more than $\lambda$, then we define $\lambda$ as the robustness of mechanism $f$. We refer to the mechanism $f$ as $\gamma$-consistent, $\lambda$-robust.

\begin{theorem}
    If $f$ is a nice mechanism with prediction applied to the two-facility game on a line, and its robustness is bounded, then its consistency is $\Omega(n)$.
\end{theorem}



\begin{proof}
    All the nice mechanisms with prediction applied to the two-facility game on a line need to comply with Theorem 1 since they are a special kind of mechanism for this problem.

    When the number of agents $n\geq5$ and $n$ is even, consider the location profile shown in Figure 3, where the locations $\epsilon$ and $1-\epsilon$ both have $n/2-1$ agents and the locations $0$ and $1$ both have an agent. We know that $\min{\mathbf{x}}=0$, $\max{\mathbf{x}}=1$. Let $\varepsilon\in(0,1/4)$ which is sufficiently small. We can see that $OPT(\mathbf{x})=2\varepsilon$, where the optimal facility location is $\varepsilon$ and $1-\varepsilon$.

    \begin{figure}[H]
    \centering
    \includegraphics[width=0.5\textwidth]{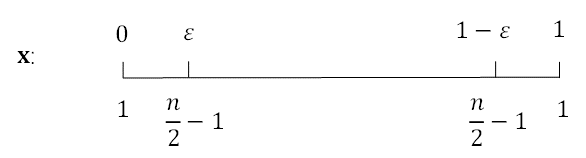}
    \caption{location profile $\mathbf{x}$. The labels below the line segment indicate the number of agents located at this position, while the labels above the line segment indicate the positions of the points.}
    \label{figure3-3}
    \end{figure}

    According to Theorem 1, let $f(\mathbf{x})=\{l_1,l_2\}$, where either $l_1=l_2$ or there exists $l\in f(\mathbf{x})$ such that $l\leq0$ or $l\geq1$. If $l_1=l_2$, there are at least $\frac{n}{2}$ agents have a cost greater than $\varepsilon$, then $SC(f,\mathbf{x})\geq\frac{n}{2}\varepsilon$, which implies approximate ratio is at least $\frac{n}{4}$. If there exists $l\in f(\mathbf{x})$ such that $l\leq0$ or $l\geq1$, then we have $SC(f,\mathbf{x})\geq\frac{n}{2}\varepsilon$, which implies a approximate ratio of at least $\frac{n}{4}$.

    When the number of agents $n\geq5$ and $n$ is odd, just let $\frac{n-3}{2}$ agents located in $\varepsilon$ and $\frac{n-1}{2}$ in $1-\varepsilon$, and the locations $0$ and $1$ both have an agent, this will not affect the result.

    For a mechanism with prediction, when the prediction is accurate, it means $\widetilde{\mathbf{x}}=({\widetilde{x}}_1,...,{\widetilde{x}}_n)$ given by an oracle is the same as $\mathbf{x}$, regardless of how this mechanism utilizes the prediction, the approximate ratio is at least $\frac{n}{4}$. This means the consistency of any mechanism with prediction is at least $\frac{n}{4}$, which is $\Omega(n)$.
    
    Theorem 2 is proven. 
\end{proof}
    
It follows from the result in this section that it is impossible to obtain a sublinear consistency by improving the mechanism with prediction proposed in \cite{xu2022mechanism}, which achieves $(1+\frac{n}{2})$-consistency and $(2n-1)$-robustness with prediction.
It is important to highlight that the use of prediction to improve the bad approximation ratio is applicable to many other mechanism design problems. Therefore, it is a strong negative result to show that prediction can not effectively improve the consistency of the mechanism for this problem. 

\section{Strategy-Proof Mechanisms}
The mechanism proposed in \cite{procaccia2013approximate} will put two facilities on the leftmost agent's location and the rightmost agent's location, where the approximation ratio is $n-2$. Mechanism 1, given below and first presented in \cite{lu2010asymptotically}, is a mechanism with a dictator that achieves a (2n-1) approximation ratio \cite{xu2022mechanism}.
These are the only two nice mechanisms we know for this problem. 

We first point out that the approximation ratio of Mechanism 1 is actually $n-1$, which is just a trivial extension of \cite{lu2010asymptotically}, and a small improvement on the approximation ratio of $(2n-1)$ \cite{xu2022mechanism}. See the proof in Appendix A. The lower bound of the approximation ratio for a nice mechanism is $\left(n-2\right)$ \cite{fotakis2014power}. Therefore, the approximation ratio of the mechanism proposed in \cite{lu2010asymptotically} is almost optimal.

\begin{mechanism}[h]
    \caption{\cite{lu2010asymptotically}}
    \label{mech:1}
    \begin{algorithmic}[1]
        \REQUIRE The location profile reported by agents $\mathbf{x}=(x_1,...,x_n)$ and a dictator $t$
        \ENSURE The locations of the two facilities
        
        \STATE Set $l_1\gets x_{t}$

        \STATE Define $d_A=\max\limits_{x_j\leq x_{t}}(x_{t}-x_j), d_B=\max\limits_{x_j\geq x_{t}}(x_j-x_{t})$

        \IF {$d_A\le d_B$}
        \STATE $l_2\gets l_1+\max\{2d_A,d_B\}$.
        \ELSE
        \STATE $l_2\gets l_1-\max\{d_A,2d_B\}$.
        \ENDIF
        \RETURN $l_1,\ l_2$
    \end{algorithmic}
\end{mechanism}

Assuming the unique dictator is the agent t. For a normalized location profile $\mathbf{x}$, let us keep the position of all agents except the dictator. If there exists a location $a\in(0, 1)$ such that for $\varepsilon\rightarrow0$, when agent $t$ at $a-\varepsilon$ and $t$ at $a+\varepsilon$, the $l_2$ locations are respectively $l_2\leq0\ (resp.\ l_2\geq1)$ and $l_2\geq1\ (resp.\ l_2\leq0)$, then $a$ is called a switching point of mechanism. Mechanism 1 is a scale-free nice mechanism and has only one switching point at $\frac{1}{2}$.

Based on the characterization result in section 3, we propose three new classes of mechanisms for this problem that have different features in terms of the position, number, and selection method of switching points. As we construct these mechanisms, we ensure that they comply with the characterization result, and we show that they are strategy-proof and have bounded approximation ratios. These mechanisms demonstrate the diversity of mechanisms for the two-facility game problem, which can help others get a more complete picture of the characterization of this problem, and also have significant value for understanding and solving more general facility game problems. 

To make the description and analysis of the mechanism simple, we only consider the normalized location profile in the paper. According to the definition of scale-free, it is easy to generalize the mechanism to the case of a general location profile. 
\vspace{-4pt}
\subsection{Extension 1: Switching Points with Different Position}

Here we present a class of scale-free nice mechanisms with only one switching point $a\in(0,1)$. It is Mechanism 2 given above, let $k\in[2,+\infty)$ be a parameter.

\begin{mechanism}[h]
    \caption{}
    \label{mech:2}
    \begin{algorithmic}[1]
        \REQUIRE The normalized location profile reported by agents $\mathbf{x}=(x_1,...,x_n)$, the unique dictator $t$, two constant parameter $k\in[2,+\infty)$ and $a\in(0,1)$.
        \ENSURE The locations of the two facilities
        \STATE Set $l_1\gets x_t$
        \IF {$x_t\in [0,a)$}
        \STATE Set $l_2\gets x_t+\max\{\frac{(1-a)k}{a}x_t,1-x_t\}$.
        \ELSE
        \STATE Set $l_2\gets x_t-\max\{x_t,\frac{ak}{1-a}(1-x_t)\}$.
        \ENDIF
        \RETURN $l_1,\ l_2$
    \end{algorithmic}
\end{mechanism}

 In the analysis of strategy-proofness, when some agents cheat so that the leftmost or rightmost position changes, the profile is no longer a normalized location profile. So we first describe the facility location output for general location profiles by Mechanism 2 given above. For a general location profile $\mathbf{x}$, we denote the leftmost agent as $x_l$, the rightmost agent as $x_r$, and $L=x_r-x_l$. 
\begin{itemize}[leftmargin=0.5cm]
    \item Case 1: 
    If the proportion of $x_t$ in $L$ is $[0,a)$, which means $x_t\in[x_l, x_l+aL)$, then $l_2=x_t+\max\{\frac{(1-a)k}{a}(x_t-x_l),x_r-x_t\}$. 
    \item Case 2: 
    If the proportion of $x_t$ in $L$ is $[a,1]$, which means $x_t\in[x_l+aL, x_r]$, then $l_2=x_t-\max\{x_t-x_l,\frac{ak}{1-a}(x_r-x_t)\}$.
\end{itemize}

\begin{theorem}
    Mechanism 2 is strategy-proof.
\end{theorem}

\begin{proof}
    Without loss of generality, let us assume $x_t \in [0, a)$ for the given normalized location profile $\mathbf{x}$. The case where $x_t \in [a, 1]$ is similar. When $x_t \in [0, a)$, we know that $l_2 \geq 1$. When somebody misreports her location, we assume another facility location except for $x_t$ is ${l_2}^\prime$.

    The cost for the dictator $t$ is 0 in $f(\mathbf{x})$, so she cannot reduce her cost by any means.
    
    We first prove that the leftmost and rightmost agents cannot benefit from misreporting. We denote the leftmost agent as $l$ and the rightmost agent as $r$, and their reported locations are ${x_r}^\prime$ and ${x_l}^\prime$. For the reported location profile $\mathbf{x}$, the leftmost and rightmost location in $\mathbf{x}$ is $y_l$ and $y_r$. Note that $y_l$ is not necessarily ${x_l}^\prime$, because $l$ can misreport her location arbitrarily(resp.\ $r$). We know that $y_r - y_l = L$, and $L=1$ if $l$ and $r$ do not misreport. 
    
    For the leftmost agent $l$, if she misreports her location from 0 to ${x_l}^\prime$, consider the location profile $\langle\mathbf{x}_{-l}, {x_l}^\prime\rangle$. If ${x_l}^\prime>0$, then $y_l\geq 0$ and $y_r\geq 1$, $L<1$, the proportion of $x_t$ in $L$ is still $[0,a)$, ${l_2}^\prime \geq 1$, so she still needs to choose the facility in $x_t$, she cannot benefit. If ${x_l}^\prime<0$, then $y_l<0$ and $y_r=1$, $L>1$. When the proportion of $x_t$ in $L$ is still $[0, a)$, ${l_2}^\prime \geq 1$, she cannot benefit. When the proportion of $x_t$ in $L$ is $[a, 1]$, ${l_2}^\prime \leq 0$. Using the formula in Case 2, note that $x_r$ in the formula is $y_r=1$ and $x_l$ in the formula is $y_l={x_l}^\prime=y_r-L=1-L$. The cost for agent $l$ is 
    \begin{align*}
        &|x_t - \max\{x_t-(1-L), \frac{ak}{1-a}(1-x_t)\}| \\\geq &\max\{x_t, \frac{ak}{1-a}(1-a)\} - x_t
        \geq ak - x_t \geq a
    \end{align*}
    the validity of the first inequality is ensured by two conditions: firstly, $x_t-(1-L)=x_t+L-1>x_t$ holds because $L>1$, and secondly, $1-x_t>1-a$ holds because $x_t<a$. The second inequality holds because $x_t \in [0, a)$ and $k \geq 2$. Since $l_1 = x_t \in [0, a)$, the leftmost agent is closer to $l_1$ than ${l_2}^\prime$, so she cannot benefit.
    
    For the rightmost agent $r$, if she misreports her location from 1 to ${x_r}^\prime$, consider the location profile $\langle\mathbf{x}_{-r}, {x_r}^\prime\rangle$. If ${x_r}^\prime>1$, then since $y_r>1$ and $L>1$, the proportion of $x_t$ in $L$ is still $[0, a)$. Using the formula in Case 1, note that $x_l$ in the formula is $y_l=0$ and $x_r$ in the formula is $y_r={x_r}^\prime=L-{x_l}^\prime=L$. We have:
    \begin{align*}
    {l_2}^\prime &= x_t + \max\{\frac{(1-a)k}{a}x_t,L-x_t\} \\&\geq x_t + \max\{\frac{(1-a)k}{a}x_t,1-x_t\} = l_2
    \end{align*}
    so she cannot benefit. If ${x_r}^\prime<1$, then $y_r<1$. When the proportion of $x_t$ in $L$ is still $[0, a)$, using the formula in Case 1, note that $x_l$ in the formula is $y_l=0$ and $x_r$ in the formula is $y_r<1$. The cost for her is 
    $$|{l_2}^\prime - 1|= |x_t + \max\{\frac{(1-a)k}{a}x_t,y_r-x_t\} - 1|$$
    if $\frac{(1-a)k}{a}x_t\geq 1-x_t\geq y_r-x_t$, we have:
    \begin{align*}
        |{l_2}^\prime - 1|&\geq |x_t + \max\{\frac{(1-a)k}{a}x_t,y_r-x_t\} - 1|\\&=|x_t + \frac{(1-a)k}{a}x_t - 1|=|l_2-1|
    \end{align*}
    if $\frac{(1-a)k}{a}x_t< 1-x_t$, based on Mechanism 2 we know $l_2=1$, which is at the location of $r$, since she cannot benefit. When the proportion of $x_t$ in $L$ is $[a,1]$, based on Mechanism 2 we know ${l_2}^\prime \leq 0$, and she cannot benefit. Now we know that $l,r$ cannot benefit by misreporting.

    Then we prove that the remaining agents, except for the leftmost, rightmost agents, and the dictator, cannot benefit by misreporting. If agent $i$ is one of them, and she misreports her location from $x_i$ to ${x_i}^\prime$, consider the location profile $\langle\mathbf{x}_{-i}, {x_i}^\prime\rangle$. If ${x_i}^\prime \in [0,1]$, then $l_2$ will not move, so she cannot benefit. If ${x_i}^\prime>1$, it is the same as the case ${x_r}^\prime>1$. Because we have proved ${l_2}^\prime \geq l_2 \geq 1$, $i$ is closer to $l_2$ than ${l_2}^\prime$ on any location in $[0,1]$, so she cannot benefit. If ${x_i}^\prime<0$, it is the same as the case ${x_l}^\prime<0$. We have proved that if the proportion of $x_t$ in $L$ is $[a, 1]$, then ${l_2}^\prime \leq -a \leq 0$, and $i$ is closer to $x_t$ than ${l_2}^\prime$ on any location in $[0,1]$ since $x_t\in[0,a)$, so she cannot benefit. If the proportion of $x_t$ in $L$ is $[0, a)$, using the formula in Case 1, note that $x_r$ in the formula is $y_r={x_r}^\prime=1$ and $x_l$ in the formula is $y_l=y_r-L=1-L$. We have: 
    \begin{align*}
    {l_2}^\prime&=x_t+\max\{\frac{(1-a)k}{a}(x_t-(1-L)),1-x_t\} \\&\geq x_t+\max\{\frac{(1-a)k}{a}x_t,1-x_t\}=l_2\geq 1
    \end{align*}
    since $L>1$. So $i$ is closer to ${l_2}$ than ${l_2}^\prime$ on any location in $[0,1]$, she cannot benefit. Therefore, any remaining agent cannot benefit by misreporting in any case. Hence, we know that Mechanism 2 is strategy-proof.
\end{proof}

\begin{theorem}
The approximation ratio of Mechanism 2 is at most $max\{\frac{(1-a)k}{2a},\frac{ak}{2(1-a)}\}(n-1)$.
\end{theorem}

\begin{proof}

For a set of agents $N$ and a location profile $\mathbf{x}$, we assume that the facility locations in $\mathbf{x}$ with respect to $OPT$ are ${OPT}_1$ and ${OPT}_2$, without loss of generality, we assume that ${OPT}_1<{OPT}_2$. Let $\alpha$ be the set of agents that are closer to ${OPT}_1$ than ${OPT}_2$, and let $\beta$ be the set of agents that are closer to ${OPT}_2$ than ${OPT}_1$. Since the locations are on a line, the leftmost agent at $x_l=0$ will go to ${OPT}_1$, i.e., in the set $\alpha$, and the rightmost agent at $x_r=1$ will go to ${OPT}_2$, i.e., in the set $\beta$. If an agent has the same distance with ${OPT}_1$ and ${OPT}_2$, we stipulate that she is in set $\alpha$. Let $I_\alpha$ be the shortest line segment that covers all the locations of the agents in set $\alpha$, and let $I_\beta$ be the shortest line segment that covers all the locations of the agents in set $\beta$. Let $|I_\alpha|$ be the length of $I_\alpha$ and $|I_\beta|$ be the length of $I_\beta$. Since each agent must go to one facility, it follows that $\alpha \cup \beta = N$. Moreover, since no agent can go to both facilities, we have $\alpha \cap \beta = \emptyset$, and ${I}_\alpha$ and $I_\beta$ do not overlap. Therefore, we have $OPT(\mathbf{x}) \geq |I_\alpha| + |I_\beta|$.

We know that $a\in(0,1)$. Without loss of generality, we assume that the unique dictator $t$ in Mechanism 2 is in $\alpha$, i.e., $t\in\alpha$, and $x_t\leq |I_\alpha|$. For the case when $t\in\beta$, the analytical processes are similar to the case when $a^\prime=1-a$ and $t\in\alpha$. We discuss the approximation ratio in three cases as follows:


\begin{itemize}[leftmargin=0.5cm]
  
  \item Case 1: When $x_t\in[0,\frac{a}{(1-a)k+a})$, Mechanism 2 outputs the dictator location $x_t$ and another facility location $l_2=1$ and $l_2\in I_\beta$. For all agents $i\in\alpha$, let them choose the facility located at $x_t$ (if some of them choose $l_2$, their cost is no more than their distance to $x_t$, the same applies in the following proof), with a corresponding cost of $cost(f(\mathbf{x}),x_i)\leq |I_\alpha|\leq OPT(\mathbf{x})$. Conversely, let all agents in $\beta$ choose the facility located at $l_2$, with a corresponding cost of $cost(f(\mathbf{x}),x_i)\leq |I_\beta|\leq OPT(\mathbf{x})$. The sum of their costs is:
  
  \[\sum_{i\in\beta} cost(f(\mathbf{x}),x_i) \leq (n-1)OPT(\mathbf{x})\]
  
  \item Case 2: When $x_t\in[\frac{a}{(1-a)k+a},\frac{2a}{(1-a)k+2a})$, we know $x_t<a$ since $k\geq 2$. Mechanism 2 outputs the dictator location $x_t$ and another facility location $l_2=x_t+\frac{(1-a)k}{a}x_t$ satisfies $1\leq l_2\leq 1+(1-x_t)$, which is available. 
  For all agents $i\in\alpha$, let them choose the facility located at $x_t$, with a corresponding cost of $cost(f(\mathbf{x}),x_i)\leq |I_\alpha|\leq OPT(\mathbf{x})$. Conversely, let all agents in $\beta$ choose the facility near them, the distance from an agent $i\in \beta$ to the facility is no more than half of the distance between $x_t$ and $l_2$, i.e., $cost(f(\mathbf{x}),x_i)\leq \frac{l_2-x_t}{2} = \frac{(1-a)k}{2a}x_t \leq \frac{(1-a)k}{2a}|I_\alpha|\leq \frac{(1-a)k}{2a}OPT(\mathbf{x})$. As the cost of $t$ is 0, and the overall cost is no less than $OPT(\mathbf{x})$, the overall cost is given by:
  \[\sum_{i\in N} cost(f(\mathbf{x}),x_i)\leq \max\{\frac{(1-a)k}{2a}(n-1),1\}OPT(\mathbf{x})\]
  
  \item Case 3: When $x_t\in[\frac{2a}{(1-a)k+2a},1)$, let all agents choose the facility located at $x_t$. For agents $i\in\alpha$, their cost is $cost(f(\mathbf{x}),x_i)\leq |I_\alpha|\leq OPT(\mathbf{x})$. For all agents $i\in\beta$, their cost is $cost(f(\mathbf{x}),x_i)\leq 1-x_t\leq 1-\frac{2a}{(1-a)k+2a} \leq \{(1-\frac{2a}{(1-a)k+2a})/\frac{2a}{(1-a)k+2a}\}|I_\alpha| \leq \frac{(1-a)k}{2a}OPT(\mathbf{x})$ since $x_t\in I_\alpha$. The cost of $t$ is 0, and the overall cost is no less than $OPT(\mathbf{x})$, so the total cost is:
  
  \[\sum_{i\in N} cost(f(\mathbf{x}),x_i)\leq\max\{\frac{(1-a)k}{2a}(n-1),1\}OPT(\mathbf{x})\]
\end{itemize}

Now we have proved that the approximation ratio for the social cost of Mechanism 2 is $\max\{\frac{(1-a)k}{2a}(n-1),1\}$ when $x_t\in \alpha$, and when $x_t\in\beta$, similarly we can prove the approximation ratio for social cost is $\max\{\frac{ak}{2(1-a)}(n-1),1\}$, it can be easily proved by equivalently transform to the case when $a^\prime=a-1$, and $t\in\alpha$. Finally, the social cost is $\max\{\frac{(1-a)k}{2a},\frac{ak}{2(1-a)}\}(n-1)$. We notice that when $a=\frac{1}{2}$ and $k=2$, the approximation ratio is $n-1$, this mechanism is Mechanism 1.
\end{proof}

\vspace{-3pt}

\subsection{Extension 2: Multiple Switching Points}

In section 6, we presented a set of scale-free nice mechanisms, each with a unique switching point, by naturally extending Mechanism 1. We now show that not all scale-free nice mechanisms have a unique switching point, and in fact, a scale-free nice mechanism can have an infinite number of switching points.

We provide an example of a class of scale-free nice mechanisms that can have an infinite number of switching points.
In Mechanism 3 given above, when the location of dictator $x_t\in(\varepsilon,1-\varepsilon)$, $l_2$ can be located within either of two intervals, allowing for the possibility of any number of switching points.
We can see that in this case the facility $l_2$ is unavailable since the distance from any agent to $x_t$ is smaller than $1$, and any agent is always farther from $l_2$ than $x_t$.

\begin{mechanism}[h]
    \caption{}
    \label{mech:3}
    \begin{algorithmic}[1]
        \vspace{-1pt}
        \REQUIRE The normalized location profile reported by agents $\mathbf{x}=(x_1,...,x_n)$, the unique dictator $t$ and a constant parameter $\varepsilon\in(0,\frac{1}{2})$.
        \ENSURE The locations of the two facilities
        \STATE Set $l_1\gets x_t$
        \IF {$x_t\in [0,\varepsilon]$}
        \STATE Set $l_2\gets x_t+\max\{1-x_t,(\frac{2}{\varepsilon}-2)x_t\}$.
        \IF {$x_t\in [1-\varepsilon,1]$}
        \STATE Set $l_2\gets x_t-\max\{x_t,(\frac{2}{\varepsilon}-2)(1-x_t)\}$.
        \ELSE
        \STATE Arbitrarily pick $l_2$ from $l_2\in (-\infty,-1)\cup(2,+\infty)$. 
        \ENDIF
        \ENDIF
        \RETURN $l_1,\ l_2$
    \end{algorithmic}
    \vspace{-1pt}
\end{mechanism}

A general location profile should be normalized before inputting to mechanism 3, and we recover this location profile and the output result. We will describe the facility location of mechanism 3 output for general location profiles. For a general location profile $\mathbf{x}$, we denote the leftmost agent as $x_l$, the rightmost agent as $x_r$, and $L=x_r-x_l$. 
\begin{itemize}[leftmargin=0.5cm]
    \item Case 1: If the proportion of $x_t$ in $L$ is $[0,\varepsilon]$, which means $x_t\in[x_l, x_l+\varepsilon L]$, then $l_2=x_t+\max\{(\frac{2}{\varepsilon}-2)(x_t-x_l),x_r-x_t\}$. 
    \item Case 2: If the proportion of $x_t$ in $L$ is $(\varepsilon,1-\varepsilon)$, which means $x_t\in(x_l+\varepsilon L, x_l+(1-\varepsilon)L)$, then $l_2\in(-\infty,-L)\cup(2L,+\infty)$. 
    \item Case 3: If the proportion of $x_t$ in $L$ is $[1-\varepsilon,1]$, which means $x_t\in[x_l+(1-\varepsilon)L, x_r]$, then $l_2=x_t-\max\{x_t-x_l,(\frac{2}{\varepsilon}-2)(x_r-x_t)\}$.
\end{itemize}

\begin{theorem}
    Mechanism 3 is strategy-proof.
\end{theorem}

\begin{proof}
    
    The cost for the dictator $t$ is 0 in $f(\mathbf{x})$, so she cannot reduce her cost by any means.

    Without loss of generality, let us assume $x_t \in [0,\varepsilon]$ or $x_t\in(\varepsilon,1-\varepsilon)$ for the given normalized location profile $\mathbf{x}$. The case where $x_t \in [1-\varepsilon,1]$ is similar to $x_t \in [0,\varepsilon]$. When $x_t \in [0,\varepsilon]$, we know that $l_2 \geq 1$. When somebody misreports her location, we assume another facility location except for $x_t$ is ${l_2}^\prime$. 

    We first prove that the leftmost and rightmost agents cannot benefit from misreporting. We denote the leftmost agent as $l$ and the rightmost agent as $r$, and their reported locations are ${x_r}^\prime$ and ${x_l}^\prime$. For the reported location profile $\mathbf{x}$, the leftmost and rightmost location in $\mathbf{x}$ is $y_l$ and $y_r$. Note that $y_l$ is not necessarily ${x_l}^\prime$, because $l$ can misreport her location arbitrarily(resp.\ $r$). We know that $y_r - y_l = L$, and $L=1$ if $l$ and $r$ do not misreport. 
    
    \begin{itemize}[leftmargin=0.5cm]
    \item Case 1: if $x_t\in[0,\varepsilon]$, we know that $l_2\geq 1$. 
    
    For the leftmost agent $l$, if she misreports her location from 0 to ${x_l}^\prime$, consider the location profile $\langle\mathbf{x}_{-l}, {x_l}^\prime\rangle$.
    
    \begin{itemize}[leftmargin=0cm]
    
        \item Subcase 1.1: If ${x_l}^\prime>0$, then $y_l>0$, $L<1$, the proportion of $x_t$ in $L$ is still $[0,\varepsilon]$, so ${l_2}^\prime \geq 1$, she cannot benefit. 
        
        \item Subcase 1.2: If ${x_l}^\prime<0$, then $y_l<0$, $L>1$. When the proportion of $x_t$ in $L$ is still $[0,\varepsilon]$, we have ${l_2}^\prime =x_t+\max\{(\frac{2}{\varepsilon}-2)(x_t-y_l),y_r-x_t\}$, since $y_r=x_r=1$ and $y_l=y_r-L=1-L<0$, we have ${l_2}^\prime =x_t+\max\{(\frac{2}{\varepsilon}-2)(x_t-(1-L)),1-x_t\}\geq l_2=x_t+\max\{(\frac{2}{\varepsilon}-2)x_t,1-x_t\}$, she cannot benefit. When the proportion of $x_t$ in $L$ is $(\varepsilon,1-\varepsilon)$, ${l_2}^\prime$ is unavailable, so she cannot benefit. When the proportion of $x_t$ in $L$ is $[1-\varepsilon,1]$, ${l_2}^\prime=x_t-\max\{x_t-(1-L),(\frac{2}{\varepsilon}-2)(1-x_t),x_t-(1-L)\}$ (as $y_r=1$ and $y_l=1-L$), and $x_l=0$, the distance from $x_l$ to ${l_2}^\prime$ is:
        \begin{align*}
            d(x_l,{l_2}^\prime)&=|x_t-\max\{x_t-(1-L),(\frac{2}{\varepsilon}-2)(1-x_t)\}| \\&
            \geq (\frac{2}{\varepsilon}-2)(1-x_t) - x_t 
            \geq(\frac{2}{\varepsilon}-2)(1-\varepsilon) - \varepsilon \geq \varepsilon
        \end{align*}
        since $x_t \in [0,\varepsilon]$ and $\varepsilon \in [0,\frac{1}{2}]$. Because $l_1 = x_t \in [0,\varepsilon]$, $d(l_1,x_l)\leq d({l_2}^\prime,x_l)$, the leftmost agent cannot benefit.
    \end{itemize}
    
    For the rightmost agent $r$, if she misreports her location from 1 to ${x_r}^\prime$, consider the location profile $\langle\mathbf{x}_{-r}, {x_r}^\prime\rangle$. 

    \begin{itemize}[leftmargin=0cm]
        \item Subcase 1.3:
        If ${x_r}^\prime>1$, then $y_r>1$ and $L>1$, the proportion of $x_t$ in $L$ is still $[0,\varepsilon]$, since $y_r=L>1$ and $y_l=0$,
        $${l_2}^\prime = x_t+\max\{y_r-x_t,(\frac{2}{\varepsilon}-2)(x_t-y_l)\} \geq x_t+\max\{1-x_t,(\frac{2}{\varepsilon}-2)x_t\} = l_2$$
        so she cannot benefit. 
        
        \item Subcase 1.4:
        If ${x_r}^\prime<1$, then $y_r<1$ and $L<1$. When the proportion of $x_t$ in $L$ is still $[0,\varepsilon]$, the cost for her is 
        $|{l_2}^\prime - 1| = |x_t+\max\{L-x_t,(\frac{2}{\varepsilon}-2)x_t\} - 1|$ (as mentioned in Subcase 1.3). 
        When $(\frac{2}{\varepsilon}-2)x_t\leq 1-x_t$, $l_2=1$, she cannot benefit from misreporting. When $(\frac{2}{\varepsilon}-2)x_t\geq 1-x_t \geq L-x_t$, we have:
        $$|{l_2}^\prime - 1| = |x_t+(\frac{2}{\varepsilon}-2)x_t - 1|=|l_2-1|$$

        since she cannot benefit. When the proportion of $x_t$ in $L$ is $(\varepsilon,1-\varepsilon)$, ${l_2}^\prime$ is unavailable, she cannot benefit. When the proportion of $x_t$ in $L$ is $[1-\varepsilon,1]$, ${l_2}^\prime\leq0$, she cannot benefit.
    \end{itemize}
    
    \item Case 2: if $x_t\in(\varepsilon,1-\varepsilon)$, we know that $l_2$ is unavailable.
      \begin{itemize}[leftmargin=0cm]
          \item Subcase 2.1:If ${x_l}^\prime>0$, then $y_l>0$ and $L<1$. If the proportion of $x_t$ in $L$ is still $(\varepsilon,1-\varepsilon)$, ${l_2}^\prime$ is still unavailable, so she cannot benefit. If the proportion of $x_t$ in $L$ is $[0,\varepsilon]$, ${l_2}^\prime \geq 1$, so she cannot benefit. 
        
         \item Subcase 2.2: If ${x_l}^\prime<0$, then $y_l<0$ and $L>1$. When the proportion of $x_t$ in $L$ is still $(\varepsilon,1-\varepsilon)$, ${l_2}^\prime$ is still unavailable, so she cannot benefit. When the proportion of $x_t$ in $L$ is $[1-\varepsilon,1]$, ${l_2}^\prime=x_t-\max\{x_t-y_l,(\frac{2}{\varepsilon}-2)(y_r-x_t)\}$, since $y_r=1$ and $y_l=y_r-L=1-L$, the distance from $x_l$ to ${l_2}^\prime$ is:
        \begin{align*}
            d(x_l,{l_2}^\prime)&=|x_t-\max\{x_t-(1-L),(\frac{2}{\varepsilon}-2)(1-x_t)\}| \\&
            \geq (\frac{2}{\varepsilon}-2)(1-x_t) - x_t 
            \geq(\frac{2}{\varepsilon}-2)\varepsilon - (1-\varepsilon) \geq 1-\varepsilon
        \end{align*}
        since $x_t \in (\varepsilon,1-\varepsilon)$ and $\varepsilon \in [0,\frac{1}{2}]$. Because $l_1 = x_t \in (\varepsilon,1-\varepsilon)$, $d(l_1,x_l)\leq d({l_2}^\prime,x_l)$,the leftmost agent cannot benefit.
        
        \item Subcase 2.3: If ${x_r}^\prime>1$, then $y_r>1$ and $L>1$. If the proportion of $x_t$ in $L$ is still $(\varepsilon,1-\varepsilon)$, ${l_2}^\prime$ is still unavailable, so she cannot benefit. If the proportion of $x_t$ in $L$ is $(0,\varepsilon)$, ${l_2}^\prime=x_t+\max\{y_r-x_t,(\frac{2}{\varepsilon}-2)(x_t-y_l)\}$, and $y_r=L$, $y_l=0$,
        \begin{align*}
            d(x_r,{l_2}^\prime)&=x_t+\max\{L-x_t,(\frac{2}{\varepsilon}-2)x_t\}-1\\&
            \geq x_t+(\frac{2}{\varepsilon}-2)x_t-1 
            \geq\varepsilon+(\frac{2}{\varepsilon}-2)\varepsilon-1 \geq 1-\varepsilon
        \end{align*}
        since $x_t \in (\varepsilon,1-\varepsilon)$ and $\varepsilon \in [0,\frac{1}{2}]$. Because $l_1 = x_t \in (\varepsilon,1-\varepsilon)$, $d(l_1,x_r)\leq d({l_2}^\prime,x_r)$, the rightmost agent cannot benefit.

        \item Subcase 2.4: If ${x_r}^\prime<1$, then $y_r<1$ and $L<1$. If the proportion of $x_t$ in $L$ is still $(\varepsilon,1-\varepsilon)$, ${l_2}^\prime$ is still unavailable, so she cannot benefit. If the proportion of $x_t$ in $L$ is $[\varepsilon,1]$, ${l_2}^\prime \leq 0$, so she cannot benefit. 
      \end{itemize}
      We have proved that $l, r$ cannot benefit from misreporting.
    \end{itemize}
    
    Then we prove that the remaining agents, except for the leftmost, rightmost agents, and the dictator, cannot benefit from misreporting. If agent $i$ is one of them, and she misreports her location from $x_i$ to ${x_i}^\prime$, consider the location profile $\langle\mathbf{x}_{-i}, {x_i}^\prime\rangle$. 
    \begin{itemize}[leftmargin=0.5cm]
        \item Case 1: If ${x_i}^\prime \in [0,1]$, then $l_2$ will not move, so she cannot benefit. 
        \item Case 2: If ${x_i}^\prime>1$, it is the same as the case ${x_r}^\prime>1$. We have proved that when $x_t \in [0,\varepsilon]$, ${l_2}^\prime\geq l_2\geq 1$; when $x_t \in (\varepsilon,1-\varepsilon)$, ${l_2}^\prime$ is always unavailable, so agent $i$ cannot benefit.
        \item Case 3: If ${x_i}^\prime<0$, it is the same as the case ${x_l}^\prime<0$. We have proved that when $x_t \in [0,\varepsilon]$, either ${l_2}^\prime \geq l_2 \geq 1$, or ${l_2}^\prime$ is unavailable; when $x_t \in (\varepsilon,1-\varepsilon)$, ${l_2}^\prime$ is always unavailable, so agent $i$ cannot benefit.
    \end{itemize}
    Therefore, agent $i$ cannot benefit from misreporting. Hence, we proved that the remaining agents, except for the leftmost, rightmost agents, and the dictator, cannot benefit from misreporting. Mechanism 2 is strategy-proof.
\end{proof}

\begin{theorem}
    The approximation ratio of Mechanism 3 is at most $(\frac{1}{\varepsilon}-1)(n-1)$.
\end{theorem}

\begin{proof}
For a set of agents $N$ and a location profile $\mathbf{x}$, we assume that the facility locations in $\mathbf{x}$ with respect to $OPT$ are ${OPT}_1$ and ${OPT}_2$. Without loss of generality, we assume that $0\leq{OPT}_1<{OPT}_2\leq1$. Let $\alpha$ be the set of agents near ${OPT}_1$, and let $\beta$ be the set of agents near ${OPT}_2$. Since the locations are on a line, the leftmost agent at $x_l=0$ will definitely go to ${OPT}_1$, i.e., in the set $\alpha$, and the rightmost agent at $x_r=1$ will definitely go to ${OPT}_2$, i.e., in the set $\beta$. If an agent has the same distance with ${OPT}_1$ and ${OPT}_2$, we stipulate that she is in set $\alpha$. Let $I_\alpha$ be the shortest line segment that covers all the locations of the agents in set $\alpha$, and let $I_\beta$ be the shortest line segment that covers all the locations of the agents in set $\beta$. Let $|I_\alpha|$ be the length of $I_\alpha$ and $|I_\beta|$ be the length of $I_\beta$. Since each agent must go to one facility, it follows that $\alpha \cup \beta = N$. Moreover, since no agent can go to both facilities, we have $\alpha \cap \beta = \emptyset$, and ${I}_\alpha$ and $I_\beta$ do not overlap. Therefore, we have $OPT(\mathbf{x}) \geq |I_\alpha| + |I_\beta|$.

Without loss of generality, we assume that the unique dictator $t$ in mechanism 3 is in $\alpha$, i.e., $t\in\alpha$, and $x_t\leq |I_\alpha|$. For the case when $t\in\beta$, the analytical processes are similar to the case when $a^\prime=1-a$ and $t\in\alpha$.

\begin{itemize}[leftmargin=0.5cm]
  
  \item Case 1: 
  When $x_t\in[0,\varepsilon]$, mechanism 3 outputs the dictator location $x_t$ and another facility $l_2$, whose location is greater than or equal to 1. For agent $i \in\alpha$, let her choose the facility located at $x_t$ (if some of them choose $l_2$, their cost is no more than their distance to $x_t$, the same applies in the following proof), with a corresponding cost $cost(f(\mathbf{x}),x_i)\leq |I_\alpha|\leq OPT(\mathbf{x})$. Conversely, let agent $i\in\beta$ choose the facility located at $l_2$. When  $l_2=1$ and $l_2\in I_\beta$, with a corresponding cost of $cost(f(\mathbf{x}),x_i)\leq |I_\beta|\leq OPT(\mathbf{x})$. In the case where $l_2>1$, the distance from an agent to facility is no more than half of the distance between $x_t$ and $l_2$, i.e., $cost(f(\mathbf{x}),x_i)\leq \frac{x_{l_2}-x_{t}}{2}\leq (\frac{1}{\varepsilon}-1)x_t \leq (\frac{1}{\varepsilon}-1)|I_\alpha|\leq (\frac{1}{\varepsilon}-1)OPT(\mathbf{x})$. As the cost of $t$ is 0, the overall cost is given by:
  \[\sum_{i\in\mathbf{N}} cost(f(\mathbf{x}),x_i)\leq (\frac{1}{\varepsilon}-1)(n-1)OPT(\mathbf{x})\]
  
  \item Case 2: When $x_t\in(\varepsilon,1-\varepsilon)$, Mechanism 3 outputs the dictator location $x_t$ and another facility $l_2$ which is unavailable. Let all agents choose the facility located at $x_t$. For agent $i\in\alpha$, her cost is $cost(f(\mathbf{x}),x_i)\leq |I_\alpha|\leq OPT(\mathbf{x})$. For agent $i\in\beta$, her cost is $cost(f(\mathbf{x}),x_i)\leq 1-\varepsilon \leq \frac{1-\varepsilon}{\varepsilon}|I_\alpha| \leq (\frac{1}{\varepsilon}-1)OPT(\mathbf{x})$. The cost of $t$ is 0, so the total cost is:
  
  \[\sum_{i\in\mathbf{N}} cost(f(\mathbf{x}),x_i)\leq (\frac{1}{\varepsilon}-1)(n-1)OPT(\mathbf{x})\]

  \item Case 3: When $x_t\in[1-\varepsilon,1]$, Mechanism 3 outputs the dictator location $x_t$ and another facility location $l_2\leq0$. Let all agents choose the facility located at $x_t$. For agent $i\in\alpha$, her cost is $cost(f(\mathbf{x}),x_i)\leq |I_\alpha|\leq OPT(\mathbf{x})$. For agent $i\in\beta$, their cost is $cost(f(\mathbf{x}),x_i)\leq \varepsilon \leq 1-\varepsilon\leq I_\alpha \leq OPT(\mathbf{x})$. The cost of $t$ is 0, so the total cost is:
  
  \[\sum_{i\in\mathbf{N}} cost(f(\mathbf{x}),x_i)\leq (n-1)OPT(\mathbf{x})\]
\end{itemize}

So we have proved that the approximation ratio for the social cost of Mechanism 3 is $(\frac{1}{\varepsilon}-1)(n-1)$. When $\varepsilon=\frac{1}{2}$, the approximation ratio is $n-1$, Mechanism actually is Mechanism 1.
\end{proof}

\subsection{Extension 3: Switching Points Related to All Agents}

The location of $l_2$ output by the mechanisms described above only depends on the leftmost agents, the rightmost agents, and the dictator. In this section, we will provide a scale-free nice mechanism whose location of $l_2$ not only depends on the leftmost and rightmost agents, and the dictator but also the location of other agents.
For simplicity, we will first provide a deterministic mechanism whose location of $l_2$ depends on the leftmost agents, rightmost agents, the dictator, and the position of only one agent. Then we prove it is strategy-proof and has a bounded approximation ratio. Afterward, we will give a nice mechanism that the location of $l_2$ is correlated with all other agents, and the proof method of strategy-proofness and approximation ratio will be similar, but we omit the proof.

We can see that in Mechanism 4 given above, the position of $l_2$ is correlated with the relative position of the agent $i$ and dictator $t$. When $x_i\leq x_t$, Mechanism 4 is a specific case of Mechanism 2 when $k=2$ and $a\in(0,\frac{1}{2})$; and when $x_i>x_t$, it is a specific case of Mechanism 2 when $k=2$ and $a^\prime=1-a$.

\begin{mechanism}[h]
    \caption{}
    \label{mech:4}
    \begin{algorithmic}[1]
        \vspace{-1pt}
        \REQUIRE The normalized location profile reported by agents $\mathbf{x}=(x_1,...,x_n)$, the unique dictator $t$ and agent $i$, $i\neq t$.
        \ENSURE The locations of the two facilities
        \STATE Let a be a constant satisfies $a\in(0,\frac{1}{2})$
        \IF {$x_i\leq x_t$}
            \IF {$x_t\in [0,a)$}
            \STATE Set $l_2\gets x_t+\max\{\frac{2(1-a)}{a}x_t,1-x_t\}$.
            \ELSE
            \STATE Set $l_2\gets x_t-\max\{x_t,\frac{2a}{1-a}(1-x_t)\}$.
            \ENDIF
        \ELSE
            \IF {$x_t\in [0,1-a)$}
            \STATE Set $l_2\gets x_t+\max\{\frac{2a}{1-a}x_t,1-x_t\}$.
            \ELSE
            \STATE Set $l_2\gets x_t-\max\{x_t,\frac{2(1-a)}{a}(1-x_t)\}$.
            \ENDIF
        \ENDIF
        \RETURN $l_1,\ l_2$
    \end{algorithmic}
\vspace{-1pt}
\end{mechanism}
\vspace{-2pt}

\addtocounter{subsection}{1}

\begin{theorem}
    Mechanism 4 is strategy-proof.
\end{theorem}

    Theorem 3 has proved that the mechanism is strategy-proof whether $x_i\leq x_t$ or $x_i>x_t$. Now we only need to prove agent $i$ cannot benefit by misreporting. We will prove that agent $i$ cannot benefit by misreporting.

\begin{proof}

    For the agent $i$, if she misreports her location from $x_i$ to ${x_i}^\prime$, consider the location profile $\langle\mathbf{x}_{-i}, {x_i}^\prime\rangle$. 
    
    First we discuss the case that $x_i\in(0,1)$ and ${x_i}^\prime\in(0,1)$. When $x_i\leq x_t(resp.\ x_i>x_t)$ and ${x_i}^\prime\leq x_t(resp.\ {x_i}^\prime>x_t)$, we know that $l_2$ keep still, $i$ cannot benefit. Therefore, it is necessary in case that $x_i\leq x_t(resp.\ x_i>x_t)$ and ${x_i}^\prime> x_t(resp.\ {x_i}^\prime\leq x_t)$, agent $i$ cannot benefit from misreporting. Without loss of generality, we let $x\in[0, a)$. We prove that as follows: 

    \begin{itemize}[leftmargin=0.5cm]
    \item Case 1: If $x_t$ satisfies $x_i\leq x_t$, $l_2\leq0$ and $x_i>x_t$, $l_2\geq1$, $i$ cannot benefit from misreporting, because $i$ is always closer to $x_t$ than ${l_2}^\prime$.
    \item Case 2: If $x_t$ satisfies $l_2\geq1$ whether $x_i\leq x_t$ or $x_i>x_t$, when $x_i\leq x_t$, $l_2=x_t+\max\{\frac{2(1-a)}{a}x_t,1-x_t\}\geq x_t+\max\{\frac{2a}{1-a}x_t,1-x_t\}$, which is the location of $l_2$ when $x_i>x_t$. If $x_i\leq x_t$, she will choose $x_t$ rather than $l_2$ as $l_2-x_i>1-x_i>x_t-x_i$. If $x_i$ chooses $l_2$, we have $x_i>x_t$, and she cannot be closer to ${l_2}^\prime$ by misreporting $x_i^\prime<x_t$. She cannot benefit from misreporting.
    \item Case 3: If $x_t$ satisfies $l_2\leq0$ whether $x_i\leq x_t$ or $x_i>x_t$, when $x_i\leq x_t$, $l_2=x_t-\max\{x_t,\frac{2a}{1-a}(1-x_t)\}\geq x_t-\max\{x_t,\frac{2(1-a)}{a}(1-x_t)\}$, which is the location of $l_2$ when $x_i>x_t$. If $x_i\geq x_t$, she will choose $x_t$ rather than $l_2$ as $x_i-l_2>x_i>x_i-x_t$. If $x_i$ chooses $l_2$, we have $x_i<x_t$, and she cannot be closer to ${l_2}^\prime$ by misreporting $x_i^\prime>x_t$. She cannot benefit by misreporting.
    \end{itemize}

    When $x_i\in(0,1)$ and ${x_i}^\prime\leq 0$ or ${x_i}^\prime\geq 1$, it is the same as ${x_l}^\prime<0$ or ${x_r}^\prime>1$, may be accompanied with the case when $x_i<x_t(resp.\ x_i>x_t)$ and $x_i^\prime>x_t(resp.\ x_i^\prime<x_t)$.
    Subsection 6.1 has proved that whether $x_i\leq x_t$ or $x_i>x_t$, everybody cannot benefit if somebody misreports which causes the leftmost agent's location change from $x_l=0$ to ${y_l}<0$ or the rightmost agent's location change from $x_r=1$ to ${y_r}>1$. Because we have proved that agent $i$ cannot benefit if $x_i\in(0,1)$ and misreport ${x_i}^\prime\in(0,1)$, she also cannot benefit if $x_i\in(0,1)$ and ${x_i}^\prime\leq 0$ or ${x_i}^\prime\geq 1$.

    Then we discuss the case that $x_i$ is the leftmost or rightmost agent, which means $x_i\in\{0,1\}$. Without loss of generality, we assume $x_i=0$, and she misreport to ${x_i}^\prime$. We know that $x_i\leq x_t$.

    \begin{itemize}[leftmargin=0.5cm]
        \item Case 4: When ${x_i}^\prime$ is still no more than $x_t$, it is equal to ${x_i}^\prime\leq x_t$ and the leftmost agent $l$ misreport her location. Section 6 has proved $l$ cannot benefit by misreporting, and in this case, $i$ is $l$, so she cannot benefit.
        \item Case 5: When ${x_i}^\prime>x_t$, we can divide this process into three parts. 
        First, we assume that agent $i$ misreport from $x_i=0$ to ${x_i}^\prime=x_t$, in Case 4 we have just proved that $i$ cannot benefit. Second, we assume that $x_i$ misreport from $x_i=x_t$ to $x_t<{x_i}^\prime<x_r$. We know that $x_l$ and $x_r$ keep still, and we have just proved that $i$ cannot benefit from Case 1 to Case 3. Finally, we assume that $x_i$ misreport from $x_t<x_i<x_r$ to ${x_i}^\prime>x_r$. It is the same as $r$ misreport from $x_r$ to ${x_r}^\prime>1$, Subsection 6.1 has proved that all agents cannot benefit, so $i$ cannot benefit. 
        Because $i$ cannot benefit in any part, and we know that when ${x_i}^\prime>x_t$, it must include at least the first two parts. So $i$ cannot benefit by misreporting
    \end{itemize}
    
\end{proof}

\vspace{-2pt}
\begin{theorem}
    The approximation ratio for social cost of Mechanism 4 is $\frac{1-a}{a}(n-1)$
\end{theorem}
\vspace{-2pt}

\begin{proof}
Through the proof in Theorem 4, we know that the approximation ratio of Mechanism 2 is $max\{\frac{(1-a)k}{2a},\frac{ak}{2(1-a)}\}(n-1)$. When $x_i\leq x_t$, Mechanism 4 is a specific case of Mechanism 2 when $k=2$ and $a\in(0,\frac{1}{2})$; and when $x_i>x_t$, it is a specific case of Mechanism 2 when $k=2$ and $a^\prime=1-a$. Since $a\in(0,\frac{1}{2})$, we know that $\frac{(1-a)}{2a}>\frac{a}{2(1-a)}$. Then the approximation ratio of Mechanism 4 is $\frac{1-a}{a}(n-1)$.
\end{proof}
\vspace{-1pt}
\begin{theorem}
    There exists a general example where the switching point is related to all agents.
\end{theorem}
\vspace{-2pt}
The switching point of the scale-free nice mechanism can be related to the relative position of dictator $t$ and each agent except for the dictator. In other words, the switching point can be related to the permutation of agents. Mechanism 5 is an example.

\begin{mechanism}[h]
    \caption{}
    \label{mech:5}
    \begin{algorithmic}[1]
        \REQUIRE The normalized location profile reported by agents $\mathbf{x}=(x_1,...,x_n)$, a unique dictator $t$.
        \ENSURE The locations of the two facilities
        \STATE Let the leftmost agent be $l$, the rightmost agent be $r$, and set $S=N/t$
        \STATE Set $a\gets \frac{1}{2}$
        \FOR{agent $i\in S$}
            \STATE Let $c_i$ be a arbitrary number satisfies $c_i \in (0,\frac{1}{2n})$.
            \IF{$x_i\leq x_t$}
                \STATE Set $a \gets a-c_i$.
            \ELSE
                \STATE Set $a \gets a+c_i$.
            \ENDIF
        \ENDFOR
        \IF {$x_t\in [0,a)$}
            \STATE Set $l_2\gets x_t+\max\{\frac{2(1-a)}{a}x_t,1-x_t\}$.
            \ELSE
            \STATE Set $l_2\gets x_t-\max\{x_t,\frac{2a}{1-a}(1-x_t)\}$.
            \ENDIF
        \RETURN $l_1,\ l_2$
    \end{algorithmic}
\end{mechanism}

Just as we demonstrated the strategy-proofness of Mechanism 4, we can similarly prove that Mechanism 5 is also strategy-proof. Additionally, as demonstrated in Subsection 6.2, the approximation ratio of mechanism 5 is bounded.

\section{Conclusion}

In this paper, we provide a more detailed characterization result and give three different classes of mechanisms with a unique dictator. We claim that the switching point can be any number between 0 and 1, its location can be related to the permutation of agents, and a mechanism can have an arbitrary number of switching points greater than one. These results demonstrate the rich diversity of strategy-proof mechanisms, which can help others get a more complete picture of the characterization of this problem. We also prove that the consistency of the two-facility game on the line with prediction is $\Omega(n)$ if we require that the same mechanism can achieve a bounded approximation ratio in the case of unreliable prediction. This result indicates that the mechanism can never achieve $o(n)$ approximation ratios, even with accurate prediction.

The open question arising from this work is providing necessary and sufficient conditions for all nice mechanisms with a unique dictator. All characterization results we currently know are necessary conditions for nice mechanisms. Although the problem's mechanisms are more complex than expected, we still anticipate that others can completely characterize all the nice mechanisms in the two-facility game problem. It will not only give us a deeper understanding of the two-facility game problem but also provide us with a good starting point for more general facility problems.

\printbibliography
\vspace{-15pt}

\newpage
\appendix
\section{Appendix A: Proving the robustness of Mechanism 1}

For a set of agents $N$ and a location profile $\mathbf{x}$, we assume that the facility locations in $\mathbf{x}$ with respect to OPT are ${OPT}_1$ and ${OPT}_2$, without loss of generality, we assume that ${OPT}_1<{OPT}_2$. Let $\alpha$ be the set of agents near ${OPT}_1$, and let $\beta$ be the set of agents near ${OPT}_2$. Since the locations are on a line, the leftmost agent at $\min \mathbf{x}$ will definitely go to ${OPT}_1$, i.e., in the set $\alpha$, and the rightmost agent at $\max \mathbf{x}$ will definitely go to ${OPT}_2$, i.e., in the set $\beta$. Let $I_\alpha$ be the shortest line segment that covers all the locations of the agents in set $\alpha$, and let $I_\beta$ be the shortest line segment that covers all the locations of the agents in set $\beta$. Since each agent must go to one facility, it follows that $\alpha \cup \beta = N$. Moreover, since no agent can go to both facilities, we have $\alpha \cap \beta = \emptyset$, and ${I}_\alpha$ and $I_\beta$ do not overlap. Therefore, we have $OPT(\mathbf{x}) \geq I_\alpha + I_\beta$.

Without loss of generality, we assume that the unique dictator in the mechanism 1 is in $\alpha$, i.e., $\hat{t}\in\alpha$, and $d_A\leq I_\alpha$ 

\begin{itemize}[leftmargin=0.5cm]
  \item Case 1: When $I_\alpha=0$, as shown in Figure 4. Only the leftmost agent at $\min \mathbf{x}$ chooses $OPT_1$ in the optimal solution, and all other agents choose ${OPT}_2$. The location of the facility $l_2$ output by the mechanism 1 is $\max \mathbf{x}\in I_\beta$, where $\beta$ denotes the set of agents that choose ${OPT}_2$. Therefore, for all agents in $\alpha$, their cost is $0$. For all agents $i\in\beta$, let them all choose $l_2$, the sum of their costs is:
  
  \[\sum_{i\in\beta} cost(f(\mathbf{x}),x_i) \leq \left|\beta\right|I_\beta \leq (n-1)OPT(\mathbf{x})\]
  
  \begin{figure}[H]
      \centering      \includegraphics[width=0.5\textwidth]{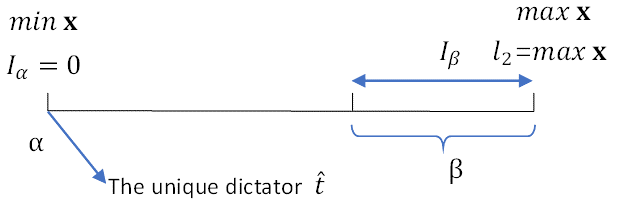}
      \caption{Case 1}
      \label{figure3-4}
  \end{figure}
  \item Case 2: 
  When ${I}_\alpha\neq0$ and $d_A\leq d_B$, as shown in Figure 5, mechanism 1 outputs the dictator location $x_{\hat{t}}$ and another facility $l_2$, whose location is greater than or equal to $\max \mathbf{x}$. For all agents $i\in\alpha$, let them choose the facility located at $x_{\hat{t}}$, with a corresponding cost of $cost(f(\mathbf{x}),x_i)\leq I_\alpha\leq OPT(\mathbf{x})$. Conversely, let all agents $i\in\beta$ choose the facility located at $l_2$. When $2d_A<d_B$, the location of $l_2$ is $\max\mathbf{x}$ and $x_{l_2}\in I_\beta$, with a corresponding cost of $cost(f(\mathbf{x}),x_i)\leq I_\beta\leq OPT(\mathbf{x})$. In the case where $2d_A\geq d_B$, the distance from an agent to facility is no more than half of the distance between $x_{\hat{t}}$ and $l_2$, i.e., $cost(f(\mathbf{x}),x_i)\leq \frac{x_{l_2}-x_{\hat{t}}}{2}\leq \frac{2d_A}{2}\leq I_\alpha\leq OPT(\mathbf{x})$. As the cost of $\hat{t}$ is 0, the overall cost is given by:
  \[\sum_{i\in\mathbf{N}} cost(f(\mathbf{x}),x_i)\leq (n-1)OPT(\mathbf{x})\]
   \begin{figure}[H]
      \centering \includegraphics[width=0.5\textwidth]{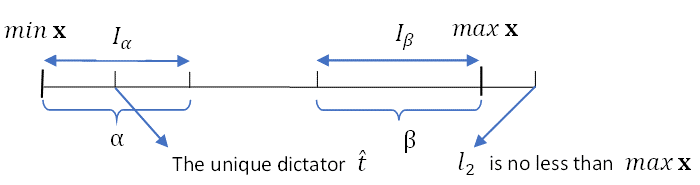}
      \caption{Case 2}
      \label{figure3-5}
  \end{figure}
  \item Case 3: When ${I}_\alpha\neq0$ and $d_A>d_B$, as shown in Figure 6, the mechanism 1 outputs the dictator location $\hat{t}$ and another facility $l_2$, whose location is less than or equal to $\min \mathbf{x}$. Let all agents choose the facility located at $x_{\hat{t}}$. For agents $i\in\alpha$, their cost is $cost(f(\mathbf{x}),x_i)\le I_\alpha\le OPT(\mathbf{x})$. For all agents $i\in\beta$, their cost is $cost(f(\mathbf{x}),x_i)\le d_B\le d_A\le I_\alpha\le OPT(\mathbf{x})$. The cost of $\hat{t}$ is 0, so the total cost is:
  
  \[\sum_{i\in\mathbf{N}} cost(f(\mathbf{x}),x_i)\leq (n-1)OPT(\mathbf{x})\]
  
  \begin{figure}[H]
      \centering
      \includegraphics[width=0.5\textwidth]{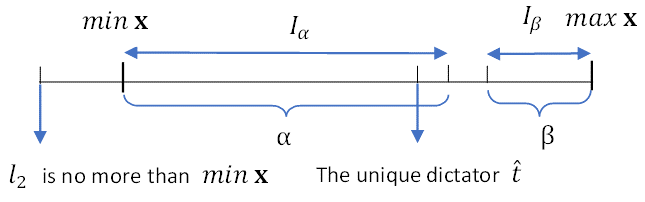}
      \caption{Case 3}
      \label{figure3-6}
  \end{figure}
\end{itemize}

In summary, the approximation ratio of the mechanism 1 is $n-1$ for two-facility game on a line.

\end{document}